\documentclass[12pt]{article}
\usepackage{amscd,amsmath,amssymb,amstext,amsthm,exscale,latexsym}
\usepackage{graphicx}
\newcommand {\dl}   {\delta}       
        
\newcommand {\ve}   {\varepsilon}  
\newcommand {\lm}   {\lambda}

\newcommand {\vf }  {\varphi}

\newcommand {\pl}   {\partial}     \newcommand {\nb}  {\nabla}
       
\renewcommand {\ln}{{\sf\,ln}}

       \renewcommand {\lim}{{\sf\,lim\,}}
\newcommand   {\ex}{{\sf\,e}}

\newcommand   {\const}{{\sf\,const}}     \newcommand   {\diag}{{\sf\,diag\,}}

\newcommand {\MO}  {{\mathbb O}}   
   \newcommand {\MR}  {{\mathbb R}}
\newcommand {\MS}  {{\mathbb S}}   
\newcommand {\MU}  {{\mathbb U}}

\newcommand {\Go}  {\mathfrak{o}}   
   
\newcommand {\Gs}  {\mathfrak{s}}   
\newcommand {\Gu}  {\mathfrak{u}}

\newtheorem{prop}{Proposition}[section]
\newtheorem{theorem}{Theorem}[section]

\theoremstyle{definition}
\begin{document}
\title     {Spherically symmetric 't Hooft--Polyakov monopoles}
\author    {M.~O.~Katanaev\thanks{E-mail: katanaev@mi-ras.ru}\\
              \sl Steklov Mathematical Institute,\\
            \sl ul. Gubkina, 8, Moscow, 119991, Russia}
\maketitle
\begin{abstract}
A general analytic spherically symmetric solution of the Bogomol'nyi equations
is found. It depends on two constants and one arbitrary function on radius and
contains the Bogomol'nyi--Prasad--Sommerfield and Singleton solutions as
particular cases. Thus all spherically symmetric 't Hooft-Polyakov monopoles
with massless scalar field and minimal energy are derived.
\end{abstract}
\section{Introduction}
The 't Hooft--Polyakov monopole solutions are exact static spherically symmetric
solutions with finite energy of the field equations of the $\MS\MU(2)$ gauge
model with the triplet of scalar fields $\vf$ in the adjoint representation and
$\lm\vf^4$ type interaction \cite{tHooft74,Polyak74}. There are many other
related solutions of the equations of motion without spherical symmetry and
different boundary conditions. All solutions are divided into homotopically
inequivalent classes parameterized by the degree of the map $\MS^2\to\MS^2$
(topological charge $Q$ taking integer values) defined by the boundary
conditions (see, e.g.\ \cite{Monast93,Rubako02,Shnir05,Katana13B}). Solutions
with spherically symmetric boundary conditions belong to the class with $Q=1$.

Monopole-type solutions have many point particle properties: finite energy,
stability and localisation in space, and are interesting both from mathematical
and physical point of view. So far they are not observed in nature. Recently,
the 't Hooft--Polyakov monopoles were given new physical interpretation in solid
state physics \cite{Katana20A,Katana21D,Katana21E} describing elastic media with
continuous distribution of disclinations and dislocations.

In each $Q$-sector of the monopole-type solutions, there are field
configurations with minimal energy. They are defined by the Bogomol'nyi
equations \cite{Bogomo76}. The solutions with minimal energy satisfy also the
original equations of motion of the model for massless scalar fields without
self interaction. Therefore solutions of the Bogomol'nyi equations are important
and interesting.

Bogomol'nyi equations reduce to the system of nonlinear ordinary differential
equations in the spherically symmetric case. The author was aware only of two
exact analytic solutions of this system of equations: the
Bogomol'nyi--Prasad--Sommerfield \cite{PraSom75,Bogomo76} and Singleton
\cite{Single95} solutions. In the present paper, we have found a general
analytic spherically symmetric solution of the Bogomol'nyi equations. It is
parameterized by one arbitrary function of radius and two constants. There is
also one degenerate solution parameterized by one arbitrary constant. In
particular cases, a general solution yields the Bogomol'nyi--Prasad--Sommerfield
and Singleton solutions. Thus we have found all spherically symmetric
't Hooft--Polyakov monopoles for massless scalar fields which minimize the
energy in the $Q=1$ sector.
\subsection{A general solution}                                   \label{smvbfy}
We consider the Euclidean space $\MR^3$ with Cartesian coordinates $x^\mu$ and
Euclidean metric $\dl_{\mu\nu}:=\diag(+++)$, $\mu,\nu=1,2,3$. Let there be the
$\MS\MU(2)$ local connection form $A_\mu{}^i(x)$ (the Yang--Mills fields) and
the triplet of scalar fields $\vf^i(x)$, $i=1,2,3$, in the adjoint
representation of $\MS\MU(2)$. The totally antisymmetric tensor is denoted by
$\ve_{ijk}$, $\ve_{123}=1$, and raising and lowering of Latin indices is
performed by the Euclidean metric $\dl_{ij}$ (the Killing--Cartan form of
$\MS\MU(2)$).

We are looking for spherically symmetric solutions of the Bogomol'nyi equations
\cite{Bogomo76}
\begin{equation}                                                  \label{unvhty}
  F_{\mu\nu}{}^i=\ve_{\mu\nu\rho}\nb^\rho\vf^i,
\end{equation}
where
\begin{equation*}
  F_{\mu\nu}{}^i:=\pl_\mu A_\nu{}^i-\pl_\nu A_\mu{}^i+A_\mu{}^j A_\nu{}^k
  \ve_{jk}{}^i
\end{equation*}
is the local curvature form (the Yang--Mills field strength) and
\begin{equation}                                                  \label{esffdd}
  \nb_\mu\vf^i:=\pl_\mu\vf^i+A_\mu{}^j\vf^k\ve_{jk}{}^i
\end{equation}
is the covariant derivative of scalar fields.

Any solution of the Bogomol'nyi equations satisfies also the field equations of
$\MS\MU(2)$ gauge model in Minkowskian space-time $\MR^{1,3}$ with massless
scalar fields without self interaction in the time gauge $A_0{}^i=0$. The
inverse statement is not true. Smooth solutions of the Bogomol'nyi equations
have minimal energy in each sector of topologically different (nonhomotopic)
solutions of the model (see, e.g.~\cite{Monast93,Rubako02,Shnir05,Katana13B}).

The boundary conditions are supposed to be spherically symmetric
\begin{equation}                                                  \label{uvcxfs}
  \underset{r\to\infty}\lim A_\mu{}^i\to0,\qquad
  \underset{r\to\infty}\lim \vf^i\to \frac{x^i}ra,\qquad a\ne0,
\end{equation}
where $r:=\sqrt{x^\mu x_\mu}$ is the usual radius in the spherical coordinate
system.

The Bogomol'nyi equations are the system of 9 first order nonlinear partial
differential equations for 12 unknown functions $A_\mu{}^i$ and $\vf^i$. They
are simpler then the original field equations of the $\MS\MU(2)$ gauge model.

Now we find a general static spherically symmetric solution of the Bogomol'nyi
equations. We assume that the global rotation group $\MS\MO(3)$ acts
simultaneously both on the base $\MR^3$, and on the Lie algebra $\Gs\Go(3)$,
which, as a vector space, is also a three-dimensional Euclidean space $\MR^3$.
It means that if $S\in\MS\MO(3)$ is an orthogonal matrix, then the
transformation has the form
\begin{equation*}
  A_\mu{}^i\mapsto S^{-1\nu}_{~~\mu} A_\nu{}^jS_j{}^i,\qquad
  S\in\MS\MO(3).
\end{equation*}
Under this assumption, the difference between Greek and Latin indices
disappears, but we shall, as far as possible, distinguish them for clarity.

The most general spherically symmetric components of the connection have the
form
\begin{equation}                                                  \label{unbcht}
  A_\mu{}^i(x):=\ve_\mu{}^{ij}\frac{x_j}rW(r)+\dl_\mu^iV(r)
  +\frac{x_\mu x^i}{r^2}U(r),
\end{equation}
where $W$, $V$, $U$ are arbitrary sufficiently smooth functions of radius.

If we include reflections into the rotation group, then $A_\mu{}^k$ are
components of the second rank pseudo-tensor with respect to the action of the
full rotation group $\MO(3)$, due to the presence of the third rank
pseudo-tensor $\ve_{ijk}$ in Eq.(\ref{esffdd}). Under the action of the full
rotation group $\MO(3)$ the function $W$ is a scalar, and $V$ and $U$ are
pseudoscalars.

The famous 't Hooft--Polyakov monopole solution \cite{tHooft74,Polyak74}
corresponds to ansatz (\ref{unbcht}) with $V\equiv U\equiv 0$.

A general spherically symmetric ansatz for the scalar fields is
\begin{equation*}
  \vf^i:=\frac{x^i}r F(r),
\end{equation*}
where $F$ is an arbitrary function.

To simplify equations, we introduce dimensionless functions $K(r)$, $L(r)$,
$M(r)$, and $H(r)$:
\begin{equation}                                                  \label{abdvvv}
  W:=\frac{K-1}{r},\qquad U:=\frac L{r},\qquad V:=\frac M{r},\qquad
  F:=\frac H{r}.
\end{equation}
Then the full system of Bogomol'nyi equations becomes
\begin{align}                                                     \label{akkgjy}
  rK'+M(L+M)=&KH,
\\                                                                \label{abcndh}
  -rK'+K^2-1-LM=&rH'-H-KH,
\\                                                                \label{absvdf}
  rM'-K(L+M)=&MH.
\end{align}

A general solution of this system of equations for $H\equiv 0$ was found in
\cite{KatVol20}, where it was given physical interpretation in solid state
physics as describing media with disclinations. Therefore we assume that $H\ne0$
in what follows.

Now we introduce new independent variable
\begin{equation}                                                  \label{ajhsjk}
  r\mapsto\xi:=\ln r,\qquad r>0,
\end{equation}
and the index will denote differentiation with respect to $\xi$, e.g.\
\begin{equation*}
  M_\xi:=\frac{dM}{d\xi}=rM',\qquad M_{\xi\xi}=r^2M''+rM'.
\end{equation*}

\begin{theorem}                                                   \label{tghqhh}
A general solution of the system of equations (\ref{akkgjy})--(\ref{absvdf}) is
\begin{align}                                                     \label{ubbxjl}
  M(\xi)=&\pm\sqrt{1-K^2+H_\xi-H},
\\                                                                \label{ujdfji}
  L(\xi)=&\mp\frac{K_\xi-K^2+1+H_\xi-(K+1)H}{\sqrt{1-K^2+H_\xi-H}},
\end{align}
where $H(\xi)$ is a solution of the Riccati equation
\begin{equation}                                                  \label{ahsgtr}
  H_\xi+H-H^2=C\ex^{2\xi}
\end{equation}
with arbitrary constant $C\in\MR$ and $K$ is an arbitrary function satisdying
inequality
\begin{equation}                                                  \label{anfjgh}
  1-K^2+H_\xi-H\ge0.
\end{equation}
The upper and lower signs in Eqs.~(\ref{ubbxjl}), (\ref{ujdfji}) must be chosen
simultaneously.
\end{theorem}

\begin{proof}
Add equations~(\ref{akkgjy}) and~(\ref{abcndh}):
\begin{equation*}
  M^2+K^2-1=rH'-H.
\end{equation*}
It implies Eq.~(\ref{ubbxjl}). Substitution of this solution into
Eq.~(\ref{abcndh}) yields
\begin{equation*}
  -K_\xi+K^2-1\mp L\sqrt{1-K^2+H_\xi-H^2}=H_\xi-H-KH.
\end{equation*}
It gives Eq.~(\ref{ujdfji}).

After substitution of $M$ (\ref{ubbxjl}) and $L$ (\ref{ujdfji})
into Eq.~(\ref{absvdf}) and changing of coordinate (\ref{ajhsjk}) all terms with
$K$ cancel, and we get the equation for $H(\xi)$:
\begin{equation}                                                  \label{anmdnf}
  H_{\xi\xi}-H_\xi-2H_\xi H-2H+2H^2=0.
\end{equation}
It is rewritten as
\begin{equation*}
  (H_\xi+H-H^2)_\xi-2(H_\xi+H-H^2)=0.
\end{equation*}
This equation can be easily integrated yielding the Riccati Eq.~(\ref{ahsgtr})
with constant of integration $C$.

We are looking for real valued solutions, therefore Eq.~(\ref{anfjgh}) must
hold.
\end{proof}

Thus we reduced the whole problem to solution of the Riccati equation
(\ref{ahsgtr}), functions $M$ and $L$ are expressed through $H$ and $K$, the
function $K$ being arbitrary.

Now we consider two special cases. Let arbitrary function $K$ satisfy equation
\begin{equation}                                                  \label{ahfjry}
  K^2=1+rH'-H.
\end{equation}
It implies $M=0$. Then Eq.~(\ref{absvdf}) yields $KL=0$, and we have two
subcases: $K=0$ and $L=0$.

{\bf Subcase $M=0$, $K=0$.} Then Eq.~(\ref{akkgjy}) is satisfied, and
Eq.~(\ref{abcndh}) yields
\begin{equation*}
  rH'-H+1=0.
\end{equation*}
Its general solution is
\begin{equation}                                                  \label{anvbfg}
  H=1+C_1r,\qquad C_1=\const.
\end{equation}
Thus we get
\begin{prop}
If Eq.~(\ref{ahfjry}) holds and $K=0$, then a general solution of the
Bogomol'nyi equations  (\ref{akkgjy})--(\ref{absvdf}) is
\begin{equation}                                                  \label{anvbfp}
  M=0,\qquad H=1+C_1r,
\end{equation}
the function $L$ being arbitrary.
\end{prop}
The gauge and scalar fields for this solution are
\begin{equation}                                                  \label{ebdvfd}
\begin{split}
  \vf^i=&\frac{x^i}r\left(\frac1r+C_1\right),
\\
  A_\mu{}^i=&-\ve_\mu{}^{ij}\frac{x_j}{r^2}+\frac{x_\mu x^i}{r^2}U(r),
\end{split}
\end{equation}
the function $U:=L/r$ being arbitrary. To satisfy the boundary conditions
(\ref{uvcxfs}) we must assume that $U(\infty)=0$ and $C_1\ne0$. This solution
seems to be new.

{\bf Subcase $M=0$, $L=0$.} Then the full system of the Bogomol'nyi equations
reduces to
\begin{equation}                                                  \label{ejgkky}
\begin{split}
  rK'=&KH,
\\
  rH'=&K^2-1+H.
\end{split}
\end{equation}
This subcase corresponds to the 't Hooft--Polyakov ansatz.
If we solve the second Eq.~(\ref{ejgkky}) for $K$ and substitute the solution
into the first equation, then we obtain Eq.~(\ref{anmdnf}). Thus the original
't Hooft--Polyakov monopoles correspond to the special case of general solution
given by Theorem~\ref{tghqhh} when arbitrary function is given by
Eq.~(\ref{ahfjry}).

{\bf A general case.} The Riccati Eq.~(\ref{ahsgtr}) in old
coordinate $r$ is
\begin{equation*}
  rH'+H-H^2=Cr^2.
\end{equation*}
Substitution
\begin{equation*}
  H(r):=\frac r{Z(r)}+1
\end{equation*}
results in the special Riccati equation (see, e.g.\ \cite{vKamke59A}, Part III,
Chapter I, Eq.~1.99)
\begin{equation}                                                  \label{anncbh}
  Z'+CZ^2=-1.
\end{equation}
Its solution going through the point $Z(0)=Z_0$ is
\begin{equation}                                                  \label{abdjuy}
  Z(r)=\begin{cases}\displaystyle
  \frac{Z_0\sqrt{-C}-\tanh(\sqrt{-C}r)}{\sqrt{-C}+CZ_0\tanh(\sqrt{-C}r)}, &
  \qquad C<0,
  \\[8pt]
  Z_0-r, &\qquad C=0,
  \\[4pt] \displaystyle
  \frac{Z_0\sqrt C-\tan(\sqrt{C}r)}{\sqrt{C}+CZ_0\tan(\sqrt{C}r)}, & \qquad C>0.
\end{cases}
\end{equation}
The constant of integration $C\ne0$ can be absorbed by rescaling the field
and radius
\begin{equation*}
  Z\mapsto\sqrt{|C|}\,Z,\qquad r\mapsto\sqrt{|C|}\,r.
\end{equation*}
Then the solution is
\begin{equation}                                                  \label{akglgi}
  Z(r)=
  \begin{cases}
    \displaystyle\frac{Z_0-\tanh r}{1-Z_0\tanh r}, &\qquad C<0,
  \\[8pt]
  Z_0-r, &\qquad C=0,
    \\[8pt]
    \displaystyle\frac{Z_0-\tan r}{1+Z_0\tan r}, &\qquad C>0.
\end{cases}
\end{equation}
For $C=0$ the solution remains the same (\ref{abdjuy}). Thus the constant $C$ in
general solution (\ref{ahsgtr}) takes, in fact, only three different values:
$C=-1,0,1$.

The scalar fields for solution (\ref{akglgi}) are
\begin{equation}                                                  \label{akglgj}
  \vf^i(r)=
  \begin{cases}
    \displaystyle\frac{x^i}r\left(\displaystyle\frac{1-Z_0\tanh r}
    {Z_0-\tanh r}+\frac1r\right),&\qquad C<0,
  \\[8pt]
  \displaystyle\frac{x^i}r\left(\displaystyle\frac1{Z_0-r}+\frac1r\right),
  &\qquad C=0,
    \\[8pt]
   \displaystyle\frac{x^i}r\left(\displaystyle\frac{1+Z_0\tan r}
   {Z_0-\tan r}+\frac1r\right), &\qquad C>0.
\end{cases}
\end{equation}
At infinity the limit is
\begin{equation*}
  \vf^i(\infty)=
  \begin{cases}
    -\displaystyle\frac{x^i}r,&\qquad C<0,
  \\
  0, &\qquad C=0,
    \\
   ?, &\qquad C>0.
\end{cases}
\end{equation*}
Thus only solutions with $C<0$ satisfy boundary condition (\ref{uvcxfs}). In the
case $C>0$ the scalar field has periodic singularities and does not have the
limit as $r\to\infty$. Therefore its physical meaning is obscure.

A general solution for the gauge field is more complicated and depends on
arbitrary function $K(r)$. We return to its analysis in future.

Up to now only a few spherically symmetric solutions of the Bogomol'nyi
equations are known.

If $C<0$, $Z_0=0$, and Eq.~(\ref{ahfjry}) holds, \emph{}then
\begin{equation}                                                  \label{abcndf}
  Z=-\tanh r,\qquad H=1-\frac r{\tanh r},\qquad K=\pm\frac r{\sinh r}.
\end{equation}
This is precisely the famous Bogomol'nyi--Prasad--Sommerfield solution
\cite{PraSom75,Bogomo76}. For $Z_0\ne0$ and arbitrary function $K(r)$ we have
infinitely many new solutions which differ, for example, by the tensorial
structure of the gauge field (\ref{unbcht}).

If $C=0$ and Eq.~(\ref{ahfjry}) holds, then
\begin{equation}                                                  \label{abcvdf}
  Z=Z_0-r,\qquad H=\frac{Z_0}{Z_0-r},\qquad K=\pm\frac r{Z_0-r}.
\end{equation}
This is the solution found in \cite{Single95}.
\section{Conclusion}
We considered the most general spherically symmetric ansatz for the gauge and
scalar fields in the $\MS\MU(2)$ gauge model. A general analytic solution of the
Bogomol'nyi equations is found. It includes the Bogomol'nyi--Prasad--Sommerfield
and Singleton solutions as particular cases. A general solution describes also
infinitely many new solutions for different values of constant $Z_0$ and
arbitrary function $K(r)$. Thus we obtained all spherically symmetric
't Hooft--Polyakov monopoles minimizing the energy in the $Q=1$ sector. All
smooth solutions have the same minimal energy.

Scalar functions for $C>0$ are not smooth. They have periodic singularities when
$r$ ranges from 0 to $\infty$ and do not have the limit as $r\to\infty$.
Therefore physical meaning of these solutions of the Bogomol'nyi equations is
obscure. Anyway, we have proved that there are no other spherically symmetric
solutions.

The Lie algebra $\Gs\Gu(2)$ is isomorphic to $\Gs\Go(3)$. Therefore the
't Hooft--Polyakov monopole solutions may be given physical interpretation in
solid state physics assuming that the rotational group $\MS\MO(3)$ acts in the
tangent space. They describe media with point disclinations
\cite{Katana20A,Katana21D,Katana21E}. Probably, the obtained spherically
symmetric solutions may be observed in solids.

This work is supported by the Russian Science Foundation under grant
19-11-00320.


\begin{thebibliography}{10}

\bibitem{tHooft74}
G.~'t~Hooft.
\newblock Magnetic monopoles in unified gauge theories.
\newblock {\em Nucl.\ Phys.\ B}, 79(2):276--284, 1974.

\bibitem{Polyak74}
A.~M. Polyakov.
\newblock Particle spectrum in the quantum field theory.
\newblock {\em JETP Letters}, 20(6):194--195, 1974.

\bibitem{Monast93}
M.~Monastyrsky.
\newblock {\em Topology of gauge fields and condenced matter}.
\newblock Springer, New York, 1993.

\bibitem{Rubako02}
V.~A. Rubakov.
\newblock {\em Classical Theory of Gauge Fields.}
\newblock Princeton University Press, Princeton, 2002.

\bibitem{Shnir05}
Ya. Shnir.
\newblock {\em Magnetic Monopoles}.
\newblock Springer--Verlag, Berlin, Heidelberg, 2005.

\bibitem{Katana13B}
M.~O. Katanaev.
\newblock Geometric methods in mathematical physics.
\newblock Ver. 4, 2020.
\newblock arXiv:1311.0733 [math-ph][in Russian].

\bibitem{Katana20A}
M.~O. Katanaev.
\newblock The 't {H}ooft--{P}olyakov monopole in the geometric theory of
  defects.
\newblock {\em Mod.\ Phys.\ Lett.}, B34(12):2050126, 2020.
\newblock https://doi.org/10.1142/S0217984920501262. arxiv: hep-th/2007.13490.

\bibitem{Katana21D}
M.~O. Katanaev.
\newblock Disclinations in the geometric theory of defects.
\newblock {\em Proc.\ Steklov.\ Inst.\ Math.}, 313:78--98, 2021.
\newblock DOI: 10.1134/S0081543821020097.

\bibitem{Katana21E}
M.~O. Katanaev.
\newblock Spin distribution for the 't {H}ooft--{P}olyakov monopole in the
  geometric theory of defects.
\newblock {\em Universe}, 7:256, 2021.
\newblock https://doi.org/10.3390/universe7080256.

\bibitem{PraSom75}
M.~K. Prasad and C.~H. Sommerfield.
\newblock Exact classical solution for the 't {H}ooft monopole and the
  {J}ulia-{Z}ee dyon.
\newblock {\em Phys.\ Rev.\ Lett.}, 35:760--762, 1975.

\bibitem{Bogomo76}
E.~B. Bogomolny.
\newblock The stability of classical solutions.
\newblock {\em Sov.\ J.\ Nucl.\ Phys.}, 24(4):449--454, 1976.

\bibitem{Single95}
D.~Singleton.
\newblock Exact Schwarzschild-like solution for Yang-Mills theories.
\newblock {\em Phys.\ Rev.\ D}, 51:5911--5914, 1995.

\bibitem{KatVol20}
M.~O. Katanaev and B.~O. Volkov.
\newblock Point disclinations in the {C}hern--{S}imons geometric theory of
  defects.
\newblock {\em Mod.\ Phys.\ Lett.}, B:2150012, 2020.
\newblock http://arxiv.org/abs/arXiv:1908.08473 [math-ph].

\bibitem{vKamke59A}
E.~von Kamke.
\newblock {\em Gew\"ohnliche Differentialgleichungen}.
\newblock Leipzig, 6 edition, 1959.

\end{thebibliography}
\end{document}